\documentclass[copyright]{eptcs}
\providecommand{\event}{DCFS 2010} 

\title{State Complexity of Catenation Combined with Star and
Reversal\thanks{This work is supported by Natural Science and
Engineering Council of Canada Discovery Grant R2824A01, Canada
Research Chair Award, and Natural Science and Engineering Council of
Canada Discovery Grant 41630. All correspondence should be directed to Bo Cui at
bcui2@csd.uwo.ca.}}
\author{Bo Cui, Yuan Gao, Lila Kari, and Sheng Yu
\institute{Department of Computer Science, \\
The University of Western Ontario,\\
London, Ontario, Canada N6A 5B7}
\email{\{bcui2,ygao72,lila,syu\}@csd.uwo.ca}
}
\def\titlerunning{State Complexity of Catenation Combined with Star and
Reversal}
\def\authorrunning{B. Cui, Y. Gao, L. Kari \& S. Yu}

\usepackage{amssymb, amsthm}

\newtheorem{theorem}{Theorem}

\newtheorem{corollary}{Corollary}
\newtheorem{lemma}{Lemma}

\usepackage{appendix}
\usepackage{amsmath}
\usepackage{graphicx}
\usepackage{latexsym}

\begin{document}
\maketitle

\begin{abstract}
This paper is a continuation of our research work on state complexity of combined operations.
Motivated by applications,
we study the state complexities of two particular combined operations:
catenation combined with star and catenation combined with reversal.
We show that the state complexities of both of these combined operations are considerably less than
the compositions of the state complexities of their individual participating operations.
\end{abstract}

\section{Introduction}
It is worth mentioning that in the past 15 years, a large number of
papers have been published on state complexities of individual
operations, for example, the state complexities of basic operations
such as union, intersection, catenation, star, etc.
\cite{DoOk09,HoKu02,JiJiSz05,Jriaskova05,PiSh2002,SaWoYu04,Yu01,YuZhSa94},
and the state complexities of several other operations such as
shuffle, orthogonal catenation, proportional removal, and cyclic
shift~\cite{CaSaYu02,DaDoSa08,Domaratzki02,JiOk05}. However, in
practice, it is common that several operations, rather than only a
single operation, are applied in a certain order on a number of
finite automata. The state complexity of combined operations is
certainly an important research direction in state complexity
research. The state complexities of a number of combined operations
have been studied in the past two years. It has been shown that the
state complexity of a combination of several operations are usually
not equal to the composition of the state complexities of individual
participating operations~\cite{GaSaYu08,JiOk07,LiMaSaYu08,SaSaYu07}.

In this paper, we study the state complexities of catenation combined with star, i.e., $L_1 L_2^*$,
and reversal, i.e., $L_1L_2^R$, respectively, where $L_1$ and $L_2$ are regular languages.
These two combined operations are useful in practice.
For example, the regular expressions that match URLs can be summarized as $L_1L_2^*$.
Also, the state complexity of $L_1L_2^R$ is equal to that of catenation combined with
{\it antimorphic involution} $(L_1 \theta(L_2))$ in biology.
An involution function $\theta$ is such that $\theta^2$ equals the identity function.
An antimorphic involution is the natural formalization of the notion of Watson-Crick complementarity
in biology.
Moreover, the combination of catenation and antimorphic involution can naturally formalize
a basic biological operation, primer extension.
Indeed, the process of creating the Watson-Crick complement of a DNA single strand $w_1 w_2$
uses the enzyme DNA polymerase to extend a known short primer $p=\theta(w_2)$
that is partially complementary to it, to obtain $\theta(w_2) \theta( w_1) = \theta (w_1 w_2)$.
This can be viewed as the catenation between the primer $p$ and $\theta(w_1)$.
The reader is referred to~\cite{Amos05} for more details about biological definitions and operations.

It has been shown in~\cite{YuZhSa94} that (1) the state complexity of
the catenation of an $m$-state DFA language (a language accepted by an $m$-state minimal complete DFA)
and an $n$-state DFA language is $m2^{n} - 2^{n-1}$,
(2) the state complexity of the star of a $k$-state DFA language, where the DFA contains
at least one final state that is not the initial state, is $2^{k-1}+2^{k-2}$,
and (3) the state complexity of the reversal of an $l$-state DFA language is $2^l$.
In this paper, we show that the state complexities of $L_1L_2^*$ and $L_1L_2^R$
are considerably less than the compositions of their individual state complexities.
Let $L_1$ and $L_2$ be two regular languages accepted by two complete DFAs of
sizes $p$ and $q$, respectively.
We will show that, if the $q$-state DFA has only one final state which is also its initial state,
the state complexity of $L_1 L_2^*$ is $p2^q - 2^{q-1}$;
in the other cases, that is when the $q$-state DFA contains some final states
that are not the initial state,
the state complexity of $L_1 L_2^*$ is $(3p - 1) 2^{q-2}$.
This is in contrast to the composition of state complexities of catenation and star
that equals $(2p - 1) 2^{2^{q-1}+2^{q-2}-1}$.
We will also show that the state complexity of $L_1L_2^R$ is $p2^q - 2^{q-1} - p + 1$
instead of $p 2^{2^q} - 2^{2^q - 1}$, the composition of state complexities of catenation and reversal.

The paper is organized as follows. We introduce the basic notations
and definitions used in this paper in the following section. Then,
we study the state complexities of catenation combined with star and
reversal in Sections~\ref{sec:cat-star} and~\ref{sec:cat-rev},
respectively. Due to page limitation, we omit the proofs of
Lemma~\ref{lem:specialcase}, Lemma~\ref{lem:cat-star-m=1},
Lemma~\ref{lem:sc-theta}, Lemma~\ref{lem:sc-cat-theta},
Theorem~\ref{L_1 L_2^R lower bound m=1 n>=3}, and Lemma~\ref{L_1
L_2^R lower bound m=1 n=2}. We also omit the proof of
Theorem~\ref{thm:cat-star-lower} for the case when $m \ge 2$ and $n
\ge 3$. We conclude the paper in Section~\ref{sec:conclusion}.

\section{Preliminaries}

An alphabet $\Sigma$ is a finite set of letters.
A word $w \in \Sigma^*$ is a sequence of letters in $\Sigma$,
and the empty word, denoted by $\lambda$, is the word of 0 length.

An involution $\theta : \Sigma \rightarrow \Sigma$ is a function such that $\theta^2 = I$
where $I$ is the identity function and can be extended to an antimorphic involution
if, for all $u,v \in \Sigma^*$, $\theta(uv) = \theta(v)\theta(u)$.
For example, let $\Sigma = \{a,b,c\}$ and define $\theta$ by $\theta(a) = b, \theta(b) = a,
\theta(c) = c$, then $\theta(aabc) = cabb$.
Note that the well-known DNA Watson-Crick complementarity is a particular antimorphic involution
defined over the four-letter DNA alphabet, $\Delta = \{A,C,G,T\}$.

A {\it non-deterministic finite automaton} (NFA) is a quintuple
$A = (Q, \Sigma, \delta, s, F)$, where $Q$ is a finite set of states,
$s \in Q$ is the start state, and $F \subseteq Q$ is the set of final states,
$\delta: Q \times \Sigma \to 2^Q$ is the transition function.
If $|\delta(q, a)| \le 1$ for any $q \in Q$ and $a \in \Sigma$, then this automaton
is called a {\it deterministic finite automaton} (DFA).
A DFA is said to be complete if $\delta(q,a)$ is defined for all $q \in Q$ and $a \in \Sigma$.
All the DFAs we mention in this paper are assumed to be complete.
We extend $\delta$ to $Q \times \Sigma^* \rightarrow Q$ in the usual way.
Then the automaton accepts a word $w \in \Sigma^*$ if $\delta(s,w) \cap F \neq \emptyset$.
Two states $p, q \in Q$ are equivalent if the following condition holds:
$\delta(p, w)\in F$ if and only if $\delta(q, w)\in F$ for all words $w\in \Sigma^*$.
It is well-known that a language which is accepted by an NFA can be accepted by a DFA,
and such a language is said to be {\it regular}.
The language accepted by a finite automaton $A$ is denoted by $L(A)$.
The reader is referred to~\cite{HoMoUl01,Yu97} for more details about regular languages
and finite automata.

The {\it state complexity} of a regular language $L$, denoted by $sc(L)$,
is the number of states of the minimal complete DFA that accepts $L$.
The state complexity of a class $S$ of regular languages, denoted by $sc(S)$,
is the supremum among all $sc(L)$, $L \in S$.
The state complexity of an operation on regular languages is the state complexity of
the resulting language from the operation as a function of the state complexities of
the operand languages.
For example, we say that the state complexity of the intersection of an $m$-state DFA language
and an $n$-state DFA language is exactly $mn$.
This implies that the largest number of states of all the minimal complete DFAs
that accept the intersection of two languages accepted by two DFAs of sizes $m$ and $n$, respectively,
is $mn$, and such languages exist.
Thus, in a certain sense, the state complexity of an operation is a worst-case complexity.

\section{Catenation combined with star}\label{sec:cat-star}
In this section, we consider the state complexity of catenation combined with star.
Let $L_1$ and $L_2$ be two languages accepted by two DFAs of sizes $m$ and $n$, respectively.
We notice that, if the $n$-state DFA has only one final state which is also its initial state, this DFA also accepts $L_2^*$.
Thus, in such a case, an upper bound for the number of states of any DFA that accepts $L_1 L_2^*  = L_1L_2$ is given by the state complexity of catenation as $m 2^{n} - 2^{n-1}$.
We first show that this upper bound is reachable by some DFAs of this form (Lemma~\ref{lem:specialcase}).
Then, we consider the state complexity of $L_1L_2^*$ in the other cases, that is when the $n$-state DFA contains some final states that are not the initial state.
We show that, in such cases, the upper bound (Theorem~\ref{thm:cat-star-upper}) coincides with the lower bound (Theorem~\ref{thm:cat-star-lower}).

\begin{lemma}\label{lem:specialcase}
For any $m \ge 2$ and $n \ge 2$, there exists a DFA $A$ of $m$
states and a DFA $B$ of $n$ states, where $B$ has only one final
state that is also the initial state, such that any DFA accepting
the language $L(A) L(B)$, which is equal to $L(A)L(B)^*$, needs at
least $m2^n-2^{n-1}$ states.
\end{lemma}

Note that, if $n = 1$, due to Theorem 3 in~\cite{YuZhSa94}, for any
DFA $A$ of size $m \ge 1$, the state complexity of a DFA accepting
$L(A)L(B)$ ($L(A)L(B)^*$) is $m$.

In the rest of this section, we only consider cases $L_1L_2^*$ where
the DFA for $L_2$ contains at least one final state that is not the
initial state. Thus, the DFA for $L_2$ is of size at least 2.

When considering the size of the DFA for $L_1$, we notice that, when
the size of this DFA is 1, the state complexity of $L_1L_2^*$ is 1.
\begin{lemma}\label{lem:cat-star-m=1}
Let $A$ be a 1-state DFA and $B$ be a DFA of $n \ge 1$ states. Then,
the necessary and sufficient number of states for a DFA to accept
$L(A)L(B)^*$ is $1$.
\end{lemma}

Now, we focus on the cases when $m > 1$ and $n > 1$, and give an upper bound for the state complexity of $L_1 L_2^*$.

\begin{theorem}\label{thm:cat-star-upper}
Let $A = (Q_1, \Sigma, \delta_1, s_1, F_1)$ be a DFA such
that $|Q_1| = m > 1$ and $|F_1| = k_1$, and $B = (Q_2, \Sigma,
\delta_2, s_2, F_2)$ be a DFA such that $|Q_2| = n > 1$ and
$|F_2 - \{s_2\}| = k_2 \ge 1$. Then, there exists a DFA of at most
$m (2^{n-1} + 2^{n-k_2 -1}) - k_1 2^{n-k_2-1}$ states that accepts
$L(A) L(B)^*$.
\end{theorem}
\begin{proof}
We denote $F_2 - \{s_2\}$ by $F_0$. Then, $|F_0| = k_2 \ge 1$.

We construct a DFA $C = \{Q, \Sigma, \delta, s, F\}$ for the language $L_1 L_2^*$, where $L_1$ and $L_2$ are the
languages accepted by DFAs $A$ and $B$, respectively.
Intuitively, $C$ is constructed by first constructing a DFA $B'$ for
accepting $L_2^*$, then catenating $A$ to this new DFA.
Note that, in the construction for $B'$, we need to add an additional initial and final state $s_2'$.
By careful examination, we can check that the states of $B'$ are state $s_2'$ and the elements
in $P - \{\emptyset\}$, where $P$ is defined in the following.
As the state set we choose
\[
    Q = \{r \cup p \mid r \in R \mbox{ and } p \in P\}, \mbox{ where }
\]
\begin{eqnarray*}
    R & = & \{ S \mid S = \{q_i\}, \mbox{ if $q_i \not\in F_1$}, S = \{q_i, s_2'\}, \mbox{ otherwise, where $q_i \in Q_1$} \}, \mbox{ and} \\
    P & = & \{S \mid S \subseteq (Q_2 - F_0)\} \cup \{ T \mid T \subseteq Q_2, s_2 \in T, \mbox{ and } T \cap F_0 \neq \emptyset\}.
\end{eqnarray*}

If $s_1 \not\in F_1$, the initial state $s$ is $s = \{s_1\} \cup
\{\emptyset\}$, otherwise, $s = \{s_1, s_2'\} \cup \{\emptyset\}$.

The set of final states $F$ is chosen to be $F = \{S \in Q \mid S \cap (F_2 \cup \{s_2'\}) \neq \emptyset\}$.

We denote a state in $Q$ as $\{q_i\} \cup G$, where $q_i \in Q_1$
and $G \subseteq Q_2 \cup \{s_2'\}$. Then, the transition relation
$\delta$ is defined as follows:
\[
    \delta(\{q_i\} \cup G, a) = D_1 \cup D_2 \cup D_3, \mbox{ for any $a \in \Sigma$}, \mbox{ where }
\]
\begin{description}
\item[$D_1$:] If $\delta_1(q_i, a) = q_i' \in F_1$, $D_1 = \{q_i', s_2'\}$, otherwise, $D_1 = \{q_i'\}$.
\item[$D_2$:] If $s_2' \in G$, then $D_2 = \{\delta_2(s_2,a)\}$, otherwise, $D_2 = \emptyset$.
\item[$D_3$:] If $G = \emptyset$, $D_3 = \emptyset$, otherwise,
\[
D_3 = \left\{
\begin{array}{l l}
  \delta_2(G,a) & \quad \text{if $\delta_2(G,a) \cap F_0 = \emptyset$,}\\
  \delta_2(G,a) \cup \{s_2\} & \quad \text{otherwise.}\\
\end{array} \right.
\]
\end{description}

We can verify that the DFA $C$ indeed accepts $L_1 L_2^*$.
It is clear that each state in $Q$ should consist of exactly one state in $Q_1$ and the states in one element of $P - \{\emptyset\}$.
Moreover, if a state of $Q$ contains a final state
of $A$, then this state also contains the state $s_2'$.

To get an upper bound for the state complexity of
catenation combined with star, we should count the number of states of $Q$.
However, as we will show in the following, some states in $Q$ are
equivalent.

Let us recall the construction for $B'$.
Note that, in that construction, states $s_2'$ and $s_2$ should reach the same state on any letter in $\Sigma$.
Also note that a state of $Q$ contains $s_2'$ only when it contains a final state of $A$.
Moreover, there exist pairs of states, denoted by $\{q_f, s_2', s_2\} \cup
T$ and $\{q_f,s_2'\} \cup T$, such that $q_f$ is a final state of $A$
and $T \subseteq Q_2 \setminus \{s_2\}$.
Then, we show that the two states in each of such pairs are equivalent as follows.
For a letter $a \in \Sigma$ and a word $w
\in \Sigma^*$,
\[
    \delta(\{q_f,s_2', s_2\} \cup T, aw) = \delta(\{q_f,s_2'\} \cup T, aw) = \delta(\delta(\{q_f,s_2'\} \cup T,a), w).
\]
Note that the equivalent states are only in the set $F_1 \times
\{s_2'\} \times \{S \mid S \subseteq (Q_2 - F_0)\}$, and we can furthermore partition this set into two sets as
\begin{eqnarray*}
& & F_1 \times \{s_2'\} \times \{s_2\} \times \{S' \mid S' \subseteq (Q_2 -  F_0 - \{s_2\})\} \cup\\
& & F_1 \times \{s_2'\}\times \{S' \mid S' \subseteq (Q_2 - F_0 -
\{s_2\})\}.
\end{eqnarray*}
It is easy to see that, for each state in the former set, there exists one and only one equivalent
state in the latter set, and vice versa.
Thus, the number of equivalent pairs is $k_1 2^{n-k_2-1}$.

Finally, we calculate the number of inequivalent states of $Q$.
Notice that there are $m$ elements in $R$, $2^{n-k_2}$ elements in
the first term of $P$, and $(2^{k_2} - 1)2^{n-k_2-1}$ elements in
the second term of $P$. Therefore, the size of $Q$ is $|Q| = m
(2^{n-1} + 2^{n-k_2 -1})$. Then, after removing one state from each
equivalent pair, we obtain the following upper bound
\[
    m (2^{n-1} + 2^{n-k_2 -1}) - k_1 2^{n-k_2-1}.
\]
\end{proof}

Next, we give examples to show that this upper bound can be reached.

\begin{theorem}\label{thm:cat-star-lower}
For any integers $m \ge 2$ and $n \ge 2$, there exists a DFA $A$ of
$m$ states and a DFA of $n$ states such that any DFA accepting
$L(A)L(B)^*$ needs at least $m \dfrac{3}{4} 2^{n} - 2^{n-2}$ states.
\end{theorem}
\begin{proof}
We first give an example of two DFAs $A$ and $B$ of sizes $m \ge 2$
and $n = 2$, respectively, and we show that the number of states of
a DFA accepting $L(A) L(B)^*$ reaches the upper bound given in
Theorem~\ref{thm:cat-star-upper}. We use a three-letter alphabet
$\Sigma = \{a,b,c\}$. We omit the cases when $n>2$, due to the page
limit.

Define $A = (Q_1, \Sigma, \delta_1, q_0, \{q_{m-1}\})$,
where $Q_1 = \{q_0, q_1, \ldots, q_{m-1}\}$, and the transitions are given as:
\begin{itemize}
\item $\delta_1(q_i, a) = q_{i+1}, i \in \{0, \ldots, m-2\}$, $\delta_1(q_{m-1},a) = q_0$,
\item $\delta_1(q_i, b) = q_{i+1}, i \in \{0, \ldots, m-3\}$, $\delta_1(q_{m-2},b) = q_0$, $\delta_1(q_{m-1},b) = q_{m-2}$,
\item $\delta_1(q_i, c) = q_{i+1}, i \in \{0, \ldots, m-3\}$, $\delta_1(q_{m-2},c) = q_0$, $\delta_1(q_{m-1},c) = q_{m-1}$.
\end{itemize}

Define $B = (Q_2, \Sigma, \delta_2, 0, \{1\})$, where $Q_2 =
\{0,1\}$, and the transitions are given as:
\begin{eqnarray*}
\delta_2(0, a) = 1, & \delta_2(0, b) = 0, & \delta_2(0, c) = 0, \\
\delta_2(1, a) = 0, & \delta_2(1, b) = 1, & \delta_2(1, c) = 0.
\end{eqnarray*}

Following the construction described in the proof of
Theorem~\ref{thm:cat-star-upper}, we construct a DFA $C =
(Q_3, \Sigma, \delta_3, s_3, F_3)$ that accepts $L(A)L(B)^*$. Note
that set $P$ only contains three elements $P = \{\emptyset, \{0\},
\{0,1\}\}$. To prove that $C$ reaches the upper bound, it is
sufficient to show that 1) all the states in $Q_3$ are reachable
from $s_3$, 2) after merging the equivalent states $\{q_{m-1}, 0'\}$ and
$\{q_{m-1}, 0', 0\}$, the remaining states are pairwise inequivalent.

We first consider the reachability of all the states. It is clear
that state $\{q_i\} \cup \{\emptyset\}$, for $i \in \{1, \ldots, m-2\}$, and state
$\{q_{m-1}, 0'\} \cup \{\emptyset\}$ are reachable from $s_3$ by reading the strings
$a^i$ and $a^{m-1}$, respectively. Then, on letters $b$ and $c$, we can reach
states $\{q_{m-2}, 0\}$ and $\{q_{m-1}, 0', 0\}$, respectively, from state
$\{q_{m-1},0'\}$. Moreover, state $\{q_i, 0\}$, $i \in
\{0, \ldots, m-3\}$, can be reached from state $\{q_{m-2}, 0\}$ by
reading the string $b^{i+1}$. Lastly, state $\{q_i,0,1\}$, $i \in
\{0, \ldots, m-2\}$, and state $\{q_{m-1}, 0',0,1\}$, are reachable
from $\{q_{m-1},0'\}$ on inputs $a^{i+1}$ and $a^{m}$, respectively.

Since states $\{q_{m-1}, 0'\}$ and $\{q_{m-1}, 0', 0\}$ are
equivalent, we remove state $\{q_{m-1}, 0', 0\}$ from $Q_3$, and
show that the rest of the states are pairwise inequivalent. Let
$\{q_i\} \cup G$ and $\{q_j\} \cup H$ be two different states in
$Q_3$ with $0 \le i \le j \le m-1$. There are three cases:

1. $i < j$. Then the string $a^{m-1-i}c$ is accepted by DFA $C$
starting from state $\{q_i\} \cup G$, but it is not accepted
starting from state $\{q_j\} \cup H$. Note that, after reading
$a^{m-1-i}c$, state $\{q_i\} \cup G$ reaches a state that contains
states $q_{m-1}$ and $0'$. In contrast, the state reached by
$\{q_i\} \cup H$ on the same input does not contain these states.
Moreover, the resulting states cannot contain state $1$, since on
letter $c$, $C$ remains in state $0$ from state $0$ and goes to state
$0$ from state $1$.

2. $i = j \neq m-1$. Since $P = \{\emptyset, \{0\}, \{0,1\}\}$ consists
of only three elements, we consider them individually. It is obvious
that, state $\{q_i, 0, 1\}$ is not equivalent to either $\{q_i\}$ or
$\{q_i,0\}$, since it is a final state but the latter two are not.
States $\{q_i\}$ and $\{q_i,0\}$ are inequivalent, since on the
string $ab$ we can reach a final state from state $\{q_i,0\}$ but not from state
$\{q_i\}$.

3. $i = j = m-1$. There are only two states $\{q_{m-1}, 0'\}$ and
$\{q_{m-1}, 0',0,1\}$. They are inequivalent, because after reading a letter
$b$, state $\{q_{m-1}, 0',0,1\}$ leads to a final state of $C$ but $\{q_{m-1},
0'\}$ does not.

Due to 1) and 2), DFA $C$ has at least $3m + 2$ pairwise inequivalent reachable states, which reaches the upper bound in Theorem~\ref{thm:cat-star-upper}.
\end{proof}

\section{Catenation combined with reversal}\label{sec:cat-rev}
In this section, we first show that the state complexity of catenation combined with an antimorphic involution $\theta$ ($L_1\theta(L_2)$) is equal to that of catenation combined with reversal.
That is, we show, for two regular languages $L_1$ and $L_2$, that $sc (L_1 \theta(L_2)) = sc(L_1 L_2^R)$ (Corollary~\ref{coro:equal-complexity}).
Then, we obtain the state complexity of $L_1L_2^R$ by proving that its upper bound (Theorem~\ref{L_1 L_2^R upper bound}) coincides with its lower bound (Theorem~\ref{L_1 L_2^R lower bound}, Theorem~\ref{L_1 L_2^R lower bound m=1 n>=3}, and Lemma~\ref{L_1 L_2^R lower bound m=1 n=2}).

We note that an antimorphic involution $\theta$ can be simulated by the composition of two simpler operations: reversal and a mapping $\phi$, which is defined as $\phi(a) = \theta(a)$ for any letter $a \in \Sigma$, and $\phi(uv) = \phi(u)\phi(v)$ where $u,v \in \Sigma^+$.
Thus, for a language $L$, we have $\theta(L) = \phi(L^R)$ and $\theta(L) = (\phi(L))^R$.
It is clear that $\phi$ is a homomorphism.
Thus, the language resulting from applying such a mapping to a regular language remains to be regular.
Moreover, we can obtain a relationship between the sizes of the two DFAs that accept $L$ and $\phi(L)$, respectively.
\begin{lemma}\label{lem:sc-theta}
Let $L \subseteq \Sigma^*$ be a language that is accepted by a minimal DFA of size $n$, $n \ge 1$. Then, the necessary and sufficient number of states of a DFA to accept $\phi(L)$ is $n$.
\end{lemma}

In order to show that the state complexity of $L_1\theta(L_2)$ is
equal to that of $L_1L_2^R$, we first show that the state complexity
of catenation combined with $\phi$ is equal to that of catenation,
i.e., for two regular languages $L_1$ and $L_2$, $sc(L_1 \phi(L_2))
= sc(L_1 L_2)$. Due to the above lemma, if $L_2$ is accepted by a
DFA of size $n$, $\phi(L_2)$ is accepted by another DFA of size $n$
as well. Thus, the upper bound for the number of states of any DFA
that accepts $L_1 \phi(L_2)$ is clearly less than or equal to $m
2^{n} - 2^{n-1}$.
The next lemma shows that this upper bound can be reached by some languages.

\begin{lemma}\label{lem:sc-cat-theta}
For integers $m \ge 1$ and $n \ge 2$, there exist languages $L_1$ and $L_2$ accepted by two DFAs of sizes $m$ and $n$, respectively, such that any DFA accepting $L_1 \phi(L_2)$ needs at least $m 2^n - 2^{n-1}$ states.
\end{lemma}

As a consequence, we obtain that the state complexity of catenation combined with $\phi$ is equal to that of catenation.

\begin{corollary}\label{coro:equal-complexity}
For two regular languages $L_1$ and $L_2$, $sc (L_1 \phi (L_2) ) = sc (L_1 L_2)$.
\end{corollary}

Then, we can easily see that the state complexity of catenation combined with $\theta$ is equal to that of catenation combined with reversal as follows.
\[
    sc( L_1 \theta(L_2)) = sc(L_1 \phi (L_2^R)) = sc (L_1 L_2^R).
\]

In the following, we study the state complexity of $L_1 L_2^R$ for regular languages $L_1$ and $L_2$.
We will first look into an upper bound of this state complexity.

\begin{theorem}
\label{L_1 L_2^R upper bound} For two integers $m,n \ge 1$, let
$L_1$ and $L_2$ be two regular languages accepted by an $m$-state
DFA with $k_1$ final states and an $n$-state DFA with $k_2$ final
states, respectively. Then there exists a DFA of at most $m2^n-k_1
2^{n-k_2}(2^{k_2}-1) -m+1$ states that accepts $L_1 L_2^R$.
\end{theorem}
\begin{proof}
Let $M=(Q_M,\Sigma , \delta_M , s_M, F_M)$ be a DFA of $m$
states, $k_1$ final states and $L_1=L(M)$. Let $N=(Q_N,\Sigma ,
\delta_N , s_N, F_N)$ be another DFA of $n$ states, $k_2$
final states and $L_2=L(N)$. Let $N'=(Q_N,\Sigma , \delta_{N'} ,
F_N, \{s_N\})$ be an NFA with $k_2$ initial states.
$\delta_{N'}(p,a)=q$ if $\delta_N(q,a)=p$ where $a\in \Sigma$ and
$p,q\in Q_N$. Clearly, $$L(N')=L(N)^R=L_2^R.$$

After performing subset construction on $N'$, we can get an
equivalent, $2^n$-state DFA $A=(Q_A,\Sigma , \delta_A , s_A, F_A)$
such that $L(A)=L_2^R$. Please note that $A$ may not be minimal and
since $A$ has $2^n$ states, one of its final state must be $Q_N$.
Now we construct a DFA $B=(Q_B,\Sigma , \delta_B , s_B, F_B)$
accepting the language $L_1 L_2^R$, where
\begin{eqnarray*}
Q_B & = & \{\langle i,j \rangle \mid i\in Q_M\mbox{, } j\in Q_A\},\\
s_B & = & \langle s_M,\emptyset \rangle, \mbox{ if } s_M \not\in F_M;\\
    & = & \langle s_M, F_N \rangle, \mbox{ otherwise}, \\
F_B & = & \{\langle i,j \rangle\in Q_B \mid j\in F_A\},\\
\delta_B(\langle i,j \rangle, a) & = & \langle i',j' \rangle \mbox{,
if } \delta_M(i,a)=i'\mbox{, }\delta_A(j,a)=j'\mbox{, }a\in
\Sigma \mbox{, } i'\notin F_M; \\
& = &  \langle i',j'\cup F_N \rangle \mbox{, if }
\delta_M(i,a)=i'\mbox{, }\delta_A(j,a)=j'\mbox{, }a\in \Sigma
\mbox{, } i'\in F_M.
\end{eqnarray*}
It is easy to see that $\delta_B(\langle i,Q_N \rangle, a) \in F_B$
for any $i\in Q_M$ and $a\in \Sigma$. This means all the states
(two-tuples) ending with $Q_N$ are equivalent. There are $m$ such
states in total.

On the other hand, since NFA $N'$ has $k_2$ initial states, the
states in $B$ starting with $i\in F_M$ must end with $j$ such that
$F_N \subseteq j$. There are in total $k_1 2^{n-k_2}(2^{k_2}-1)$ states
which don't meet this.

Thus, the number of states of the minimal DFA accepting $L_1 L_2^R$
is no more than
$$m2^n-k_1 2^{n-k_2}(2^{k_2}-1) -m+1.$$
\end{proof}

This result gives an upper bound for the state complexity of $L_1
L_2^R$.
Next we show that this bound is reachable.

\begin{theorem}
\label{L_1 L_2^R lower bound}
Given two integers $m\geq 2$, $n\geq 2$, there exists a DFA $M$ of
$m$ states and a DFA $N$ of $n$ states such that any DFA accepting
$L(M)L(N)^R$ needs at least $m2^n-2^{n-1}-m+1$ states.
\end{theorem}
\begin{proof}
Let $M=(Q_M,\Sigma , \delta_M , 0, \{m-1\} )$ be a DFA,
where $Q_M = \{0,1,\ldots ,m-1\}$, $\Sigma = \{a,b,c\}$, and the transitions are given as:
\begin{itemize}
\item $\delta_M(i, x) = i \mbox{, } i=0, \ldots , m-1, x \in \{a,b\},$
\item $\delta_M(i, c) = i+1 \mbox{ mod }m \mbox{, } i=0, \ldots , m-1.$
\end{itemize}

Let $N=(Q_N,\Sigma , \delta_N , 0, \{0\} )$ be a DFA,
where $Q_N = \{0,1,\ldots ,n-1\}$, $\Sigma = \{a,b,c\}$, and the transitions are given as:
\begin{itemize}
\item $\delta_N(0, a) = n-1 \mbox{, } \delta_N(i, a) = i-1 \mbox{, } i=1, \ldots , n-1,$
\item $\delta_N(0, b) = 1 \mbox{, } \delta_N(i, b) = i \mbox{, } i=1, \ldots , n-1,$
\item $\delta_N(0, c) = 1 \mbox{, } \delta_N(1, c) = 0\mbox{, } \delta_N(j, c) = j \mbox{, } j=2, \ldots , n-1\mbox{, if $n \geq 3.$}$
\end{itemize}

Now we design a DFA $A=(Q_A, \Sigma , \delta_A , \{0\}, F_A )$,
where $Q_A = \{q \mid q\subseteq Q_N\}$, $\Sigma = \{a,b,c\}$, $F_A
= \{q \mid 0\in q \mbox{, }q\in Q_A\}$, and the transitions are
defined as:
\[
\delta_A(p, e) = \{j \mid \delta_N(j, e)=i\mbox{, }i\in p\} \mbox{, } p\in Q_A\mbox{, } e \in \Sigma.
\]

It has been shown in~\cite{YuZhSa94} that $A$ is a minimal DFA that
accepts $L(N)^R$. Let $B=(Q_B, \Sigma = \{a,b,c\}, \delta_B , s_B  =  \langle 0,\emptyset \rangle, F_A)$ be
another DFA, where
\begin{eqnarray*}
Q_B & = & \{\langle p,q\rangle  \mid p\in Q_M-\{ m-1\} \mbox{, }q\in Q_A-\{ Q_N\}\}\cup \{\langle 0,Q_N\rangle \}\\
& & \qquad \cup \,\, \{\langle m-1, q\rangle \mid q\in Q_A-\{ Q_N\}\mbox{, }\{0\}\in q\},\\
F_B & = & \{\langle p,q\rangle \mid q\in F_A \mbox{, } \langle
p,q\rangle\in Q_B\},
\end{eqnarray*}
and for each state $\langle p,q\rangle\in Q_B$ and each letter $e\in
\Sigma,$

\begin{eqnarray*}
\delta_B(\langle p,q\rangle, e) = \left\{
\begin{array}{l l}
  \langle p',q'\rangle & \mbox{if }\delta_M(p, e)=p'\neq m-1\mbox{, } \delta_A(q, e)=q'\neq Q_N,\\
  \langle p',q'\rangle &  \mbox{if }\delta_M(p, e)=p'=m-1\mbox{, }  \\
   & \qquad \delta_A(q, e)=r'\mbox{, }\mbox{$q' =r'\cup \{0\}$, $q' \neq Q_N$,}\\
  \langle 0,Q_N\rangle & \mbox{if }\delta_M(p, e)=m-1\mbox{, } \delta_A(q, e)=r'\mbox{, $r' \cup \{0\} = Q_N$,} \\
  \langle 0,Q_N\rangle & \mbox{if }\delta_M(p, e)\neq m-1\mbox{, } \delta_A(q, e)=
  Q_N.
\end{array} \right.
\end{eqnarray*}
As we mentioned in last proof, all the states (two-tuples) ending
with $Q_N$ are equivalent. So here, we replace them with one state:
$\langle 0,Q_N\rangle$. And all the states starting with $m-1$ must
end with $j\in Q_A$ such that $0\in j$. It is easy to see that $B$
accepts the language $L(M)L(N)^R$. It has $m2^n-2^{n-1}-m+1$ states.
Now we show that $B$ is a minimal DFA.

(I) We first show that every state $\langle i,j\rangle \in Q_B$ is reachable by induction on the size of $j$. Let $k = |j|$ and $k \le n-1$.
Note that state $\langle 0, Q_N\rangle$ is reachable from state $\langle 0 , \emptyset \rangle$ over string $c^mb(ab)^{n-2}$.

When $ k = 0$, $i$ should be less than $m-1$
according to the definition of $B$.
Then, there always exists a string
$w=c^i$ such that $\delta_B(\langle 0,\emptyset \rangle, w) = \langle i,\emptyset \rangle$.

{\bf Basis ($k = 1$):} State $\langle m-1, \{0\} \rangle$ can be reached from state $\langle m-2, \emptyset \rangle$ on a letter $c$.
State $\langle 0 , \{0\} \rangle$ can be reached from state $\langle m-1, \{0\} \rangle$ on string $ca^{n-1}$.
Then, for $ i \in \{1, \ldots, m-2\}$, state $\langle i , \{0\} \rangle$ is reachable from state $\langle i-1, \{0\} \rangle$ on string $ca^{n-1}$.
Moreover, for $ i \in \{0, \ldots, m-2\}$, state $\langle i,j\rangle$ is reachable from state $\langle i, \{0\} \rangle$ on string $a^j$.

{\bf Induction steps:} Assume that all states $\langle i,j\rangle$ such that $|j| < k$ are reachable.
Then, we consider the states $\langle i,j\rangle$ where $|j| = k$.
Let $j = \{j_1, j_2, \ldots, j_k\}$ such that $ 0 \le j_1 < j_2 < \ldots < j_k \le n-1$.
We consider the following four cases:

1. $j_1 = 0$ and $j_2 = 1$.
State $\langle m-1, \{0,1,j_3,\ldots,j_k\}\rangle$ is reachable from state $\langle m-2, \{0,j_3,\ldots,j_k\}\rangle$ on a letter $c$.
Then, for $i \in \{0,\ldots, m-2\}$, state $\langle i, j\rangle$ can be reached from state $\langle m-1, \{0,1,j_3,\ldots,j_k\} \rangle$ on string $c^{i+1}$.

2. $i = 0$, $j_1 = 0$, and $j_2 > 1$.
State $\langle 0,j \rangle$ can be reached as follows:
\[
 \langle 0, \{j_1, j_2, \ldots, j_k \} \rangle = \delta_B(\langle m-2, \{j_3-j_2+1, \ldots, j_k - j_2 + 1, n - j_2 +1\} \rangle, c^2a^{j_2-1}).
\]

3. $i = 0$ and $j_1 > 0$.
State $\langle 0,j \rangle$ is reachable from state $\langle 0, \{0, j_2 - j_1, \ldots, j_k - j_1\} \rangle$ over string $a^{j_1}$.

4. We consider the remaining states.
For $ i \in \{ 1, \ldots, m-1\}$, state $\langle i,j \rangle$ such that $j_1 = 0$ and $j_2 > 1$ can be reached from state $\langle i-1, \{1, j_2, \ldots, j_k\} \rangle$ on a letter $c$, and, for $i \in \{ 1, \ldots, m-2\}$, state $\langle i,j \rangle$ such that $j_1 > 0$ is reachable from state $\langle i, \{0, j_2 - j_1, \ldots, j_k - j_1\} \rangle$ over string $a^{j_1}$.
Recall that we do not have states $\langle i ,j \rangle$ such that $i = m-1$ and $j_1 > 0$.

(II) We then show that any two different states $\langle i_1,j_1\rangle$ and $\langle i_2,j_2\rangle$ in $Q_B$ are distinguishable.
Let us consider the following three cases:

1. $j_1\neq j_2$. Without loss of generality, we may assume that
$|j_1|\geq |j_2|$. Let $x\in j_1-j_2$. We don't need to consider the
case when $x=0$, because, if $0\in j_1-j_2$, then the two states are
clearly in different equivalent classes. For $0<x\leq n-1$, there
always exists a string $t$ such that $\delta_B(\langle
i_1,j_1\rangle, t)\in F_B$ and $\delta_B(\langle i_2,j_2\rangle, t)
\notin F_B,$
where
\begin{eqnarray*}
t = \left\{
\begin{array}{l l}
  a^{n-x} & \mbox{if $i_2\neq m-1$, $j_1\neq j_2$,}\\
  a^{n-x-1}ca & \mbox{if  $i_2= m-1$, $j_1\neq j_2$, $n > 2$,}\\
  c & \mbox{if $i_2= m-1$, $j_1\neq j_2$, $n=2$.}
\end{array} \right.
\end{eqnarray*}
Note that, under the second condition, after reading the prefix $a^{n-x-1}$ of $t$, state $n-1$ cannot be in the second component of the resulting state.
This is because $x \not\in j_2$.

Also note that when $n=2$, $j_1$, $j_2\in\{Q_N,\{0\},\{1\}\}$, where
$Q_N=\{0,1\}$. Moreover, when $i_2=m-1$, $\langle i_2,j_2\rangle$
can only be $\langle m-1,\{0\}\rangle$. Due to the definition of
$B$, we have that, for $s\geq 1$, $\langle s,Q_N\rangle \notin Q_B$.
Thus, it is easy to see that $\langle i_1,j_1\rangle$ is either
$\langle i_1,\{1\}\rangle$ or $\langle 0,\{0,1\}\rangle$. When
$\langle i_1,j_1\rangle=\langle i_1,\{1\}\rangle$, $0\in j_1-j_2$,
so the two states are distinguishable. When $\langle
i_1,j_1\rangle=\langle 0,\{0,1\}\rangle$, a string $c$ can
distinguish them because $\delta_B(\langle 0,\{0,1\}\rangle, c) \in
F_B$ and $\delta_B(\langle m-1,\{0\}\rangle, c) \notin F_B$.

2. $j_1=j_2\neq Q_N$, $i_1\neq i_2$. Without loss of generality, we
may assume that $i_1> i_2$. In this case, $i_2\neq m-1$. Let $x\in
Q_N-j_1$. There always exists a string $u=a^{n-x+1}bc^{m-1-i_1}$
such that $\delta_B(\langle i_1,j_1\rangle, u)\in F_B$ and
$\delta_B(\langle i_2,j_2\rangle, u) \notin F_B$.

Let $\langle i_1,j_1'\rangle$ and $\langle i_2,j_1'\rangle$ be two
states reached from states $\langle i_1,j_1\rangle$ and $\langle
i_2,j_2\rangle$ on the prefix $a^{n-x+1}$ of $w$, respectively. We
notice that state $1$ of $N$ cannot be in $j_1'$. Then, after
reading another letter $b$, we reach states $\langle
i_1,j_1''\rangle$ and $\langle i_2,j_1''\rangle$, respectively. It
is easy to see that states $0$ and $1$ of $N$ are not in $j_1''$.
Lastly, after reading the remaining string $c^{m-1-i_1}$ from state
$\langle i_1,j_1''\rangle$, the first component of the resulting
state is the final state of DFA $M$ and therefore its second
component contains state $0$ of DFA $N$. In contrast, the second
component of the resulting state reached from state $\langle
i_2,j_1''\rangle$ on the same string cannot contain state $0$, and
hence it is not a final state of $B$. Note that this includes the
case that $j_1=j_2= \emptyset$, $i_1\neq i_2$.

3. We don't need to consider the case $j_1=j_2= Q_N$, because there is
only one state in $Q_B$ which ends with $Q_N$. It is $\langle
0,Q_N\rangle$.

Since all the states in $B$ are reachable and pairwise distinguishable, DFA
$B$ is minimal.
Thus, any DFA accepting $L(M)L(N)^R$ needs at least $m2^n-2^{n-1}-m+1$ states.
\end{proof}

This result gives a lower bound for the state complexity of
$L(M)L(N)^R$ when $m,n \ge 2$. It coincides with the upper bound when $k_1=1$ and
$k_2=1$.
In the rest of this section, we consider the remaining cases when either $m =1$ or $n=1$.
We first consider the case when $m=1$ and $n\geq 3$.
We have $L_1=\emptyset$ or $L_1=\Sigma^*$.
When $L_1=\emptyset$, for any $L_2$, a 1-state DFA
always accepts $L_1L_2^R$, since $L_1L_2^R=\emptyset$.
The following theorem provides a lower bound for the latter case.

\begin{theorem}
\label{L_1 L_2^R lower bound m=1 n>=3} Given an integer $n\geq 3$,
there exists a DFA $M$ of $1$ state and a DFA $N$ of $n$ states such
that any DFA accepting $L(M)L(N)^R$ needs at least $2^{n-1}$ states.
\end{theorem}

Now, we consider the case when $m = 1$ and $n = 2$.
\begin{lemma}
\label{L_1 L_2^R lower bound m=1 n=2}There exists a $1$-state DFA
$M$ and a $2$-state DFA $N$ such that any DFA accepting $L(M)L(N)^R$
needs at least $2$ states.
\end{lemma}

Lastly, we consider the case when $m\geq 1$ and $n=1$. When
$L_2=\emptyset$, for any $L_1$, a 1-state DFA always accepts $L_1L_2^R=\emptyset$.
When $L_2=\Sigma ^* $, $L_1L_2^R=L_1\Sigma ^*$, since $(\Sigma
^*)^R=\Sigma ^*$. Due to Theorem 3 in~\cite{YuZhSa94}, which states that, for any DFA $A$ of size $m \ge 1$, the
state complexity of $L(A)\Sigma^*$ is $m$, the following is immediate.
\begin{corollary}
\label{L_1 L_2^R lower bound m>=1 n=1} Given an integer $m\geq 1$,
there exists an $m$-state DFA $M$ and a $1$-state DFA $N$ such that
any DFA accepting $L(M)L(N)^R$ needs at least $m$ states.
\end{corollary}

After summarizing Theorems~\ref{L_1 L_2^R upper bound},~\ref{L_1 L_2^R lower bound}, and~\ref{L_1 L_2^R lower
bound m=1 n>=3}, Lemma~\ref{L_1 L_2^R lower bound m=1 n=2} and
Corollary~\ref{L_1 L_2^R lower bound m>=1 n=1}, we obtain the state
complexity of the combined operation $L_1L_2^R$.
\begin{theorem}
\label{Tight bound of L_1 L_2^R} For any integer $m\geq 1$, $n\geq
1$, $m2^n-2^{n-1}-m+1$ states are both necessary and sufficient in
the worst case for a DFA to accept $L(M)L(N)^R$, where $M$ is an
$m$-state DFA and $N$ is an $n$-state DFA.
\end{theorem}

\section{Conclusion}\label{sec:conclusion}
Motivated by their applications, we have studied the state
complexities of two particular combinations of operations:
catenation combined with star and catenation combined with reversal.
We proved that they are significantly lower than the compositions of
the state complexities of their individual participating operations.
Thus, this paper shows further that the state complexity of a
combination of operations has to be studied individually.

\section*{Acknowledgement}
We would like to thank the anonymous referees of DCFS 2010 for their careful reading and valuable
suggestions.

\bibliographystyle{eptcs} 

\end{document}